\documentclass[11pt]{article}

\usepackage[margin=1.0in]{geometry}
\usepackage[ruled]{algorithm2e} % For algorithms

\SetAlFnt{\small}
\SetAlCapFnt{\small}
\SetAlCapNameFnt{\small}
\SetAlCapHSkip{0pt}
\IncMargin{-\parindent}

\usepackage{url}
\usepackage{amstext,amssymb,amsmath}
\usepackage{amsthm} 
\usepackage{dsfont}
\usepackage[T1]{fontenc}
\usepackage[utf8x]{inputenc}
\usepackage{ifthen}
\usepackage{xspace}
\usepackage{graphicx}
\usepackage{graphicx,color}
\usepackage{fancybox}
\usepackage{lipsum}
\usepackage{array,float}
\usepackage[authoryear,square]{natbib}
\usepackage{txfonts}

\usepackage{pstricks, pst-tree, pst-node}
\usepackage{enumitem}
\usepackage{mathtools}
\usepackage[english]{babel}
\usepackage{algorithm2e}
\usepackage{pstricks}
\usepackage{comment}
\definecolor{cornellred}{rgb}{0.7, 0.11, 0.11}

\theoremstyle{definition}

\theoremstyle{definition}
\newtheorem{observation}{Observation}
\usepackage{accents}

\newtheorem{definition}{Definition}
\newtheorem{proposition}{Proposition}
\newtheorem{corollary}{Corollary}
\newtheorem{lemma}{Lemma}

\newcommand{\prob}[2][]{\mathbb{P}\ifthenelse{\not\equal{}{#1}}{_{#1}}{}\!\left[#2\right]}
\newcommand{\expect}[2][]{{\mathbb{ E}}\ifthenelse{\not\equal{}{#1}}{_{#1}}{}\!\left[#2\right]}

\newcommand{\ex}[2]{\expect[#1]{#2}}

\DeclareMathOperator{\argmax}{argmax}

\newcommand{\SSpace}{\Theta}
\newcommand{\BScheme}[1]{\phi_{\theta}(#1)}
\newcommand{\util}[2]{u_{#1}^{\theta}(#2)}
\newcommand{\Ffun}[1]{f^{\theta}(#1)}
\newcommand{\Gfun}[2]{g_{#1}^{\theta}(#2)}
\newcommand{\SDist}{\mu_{\theta}}
\newcommand{\agents}{N}
\newcommand{\pay}[2]{p^{(#1)}(#2)}
\newcommand{\pr}[2]{\prob[#1]{#2}}

\makeatletter
\makeatletter

\newcommand{\customlabel}[2]{%
   \protected@write \@auxout {}{\string \newlabel {#1}{{#2}{\thepage}{#2}{#1}{}} }%
   \hypertarget{#1}{#2}
}
\makeatother

\newenvironment{talign*}
 {\csname align*\endcsname}
 {\endalign}
\allowdisplaybreaks

\begin{document}
\title{Persuasion and Incentives Through the Lens of Duality}  

\author{
Shaddin Dughmi
\thanks{Department of Computer Science, University of Southern California, shaddin@usc.edu. Supported by NSF CAREER Award CCF-1350900.} \and
Rad Niazadeh
\thanks{Department of Computer Science, Stanford University,  rad@cs.stanford.edu.} \and 
Alexandros Psomas
\thanks{Simons Institute for the Theory of Computing, alexpsomi@cs.berkeley.edu.} \and
S. Matthew Weinberg
\thanks{Department of Computer Science, Princeton University, smweinberg@princeton.edu. Supported by NSF CCF-1717899}
}

\date{}
\maketitle

\begin{abstract}
Lagrangian duality underlies both classical and modern mechanism design. In particular, the dual perspective often permits simple and detail-free characterizations of optimal and approximately optimal mechanisms. This paper applies this same methodology to a close cousin of traditional mechanism design, one which shares conceptual and technical elements with its more mature relative: the burgeoning field of persuasion. The dual perspective permits us to analyze optimal persuasion schemes both in settings which have been analyzed in prior work, as well as for natural generalizations which we are the first to explore in depth. Most notably, we permit combining persuasion policies with payments, which serve to augment the persuasion power of the scheme. In both single and multi-receiver settings, as well as under a variety of constraints on payments, we employ duality to obtain structural insights, as well as tractable and simple characterizations of optimal policies.

\end{abstract}

\section{Introduction}
%Intro paragraph
\label{sec:intro}
There are two primary ways of influencing the actions of strategic agents: through providing incentives and through influencing beliefs. The former is the domain of traditional \emph{mechanism design}, and involves the promise of payments or goods contingent on behavior. The latter is the domain of \emph{information design}, or \emph{persuasion}, and involves the selective provision of information pertaining to the payoffs and costs of various actions. There are striking similarities and parallels  between the two worlds, both in terms of the domains in which they are studied --- for example in auctions~\citep{emeksignaling,DIR14,daskalakis_information} and routing \citep{bhaskar_signaling} --- as well in the mathematical models and techniques used to characterize and compute optimal policies (e.g. \citep{DX-16,DH-17,kolotilin_privateinfo}). Combining the approaches and techniques  of mechanism design and persuasion leads to a more powerful toolkit for the design of economic systems, and this paper takes a step in that direction.

We work with two models of persuasion: the Bayesian Persuasion model of \citet{KG-11}, and the multi-receiver private Bayesian persuasion model of \citet{AB-16} (further developed in \cite{babichenko_barman} and \citet{DH-17}) which we generalize to allow externalities.
 In the spirit of mechanism design and principal-agent problems, we generalize both models by permitting payments,  which serve as additional incentive for the receiver(s) of information to behave in accordance with the wishes of the sender (principal). 
 
 We then explore these models through the lens of \emph{Lagrangian duality}, much in the spirit of the literature applying duality to auction theory and Bayesian mechanism design. In particular, we vary constraints on the payments (arbitrary, nonnegative, budget balanced) and the information/reward structure (symmetric vs asymmetric actions), and derive canonical and/or tractable optimal policies through duality.

\paragraph{The Persuasion Models.}

 In the Bayesian persuasion model, there is a \emph{receiver} who must select one of a number of actions, and a \emph{sender} looking to influence the receiver's choice in order to maximize her own expected payoff.  We adopt the perspective of the sender. A \emph{state of nature}, drawn from a common knowledge prior distribution, determines the payoff of each action to each of the sender and the receiver. The sender has an informational advantage over the receiver: access to the realization of the state of nature. The problem facing the sender is that of computing and committing to the optimal \emph{signaling scheme}: a randomized map from states of nature to signals. Once the state is drawn by nature, the signaling scheme is invoked and the corresponding signal is sent to the receiver; she then updates her prior belief and chooses the action maximizing her expected payoff.
The multi-receiver private Bayesian persuasion model generalizes the previous model to multiple receivers. There is still a common knowledge prior distribution over states of nature, and a single sender with an informational advantage. We restrict attention to the special case of two actions $\{0,1\}$ for each receiver. The state of nature now determines a set function for the sender and a set function for each of the receivers: each set function maps the set of receivers taking action $1$ to a payoff. %For one of our results, we assume that the set functions of the receivers satisfy a natural \emph{positive externality} property which we define formally.
 A signaling scheme now is a randomized map from states of nature to a signal for each receiver.

In both models, a simple revelation principle style argument shows that it suffices to restrict attention to schemes which are \emph{direct} and \emph{persuasive} (see e.g. \cite{KG-11,AB-16}). A direct scheme is one in which signals correspond to action recommendations. Such a scheme is persuasive if it is a Bayes-Nash equilibrium for each receiver to follow the recommendation.

\vspace{-3mm}
\paragraph{Adding Payments.}
%Payments: Natural model (pay for following recommendation). Payment identity. mention cheng li paper

We augment each model by allowing a special form of payment contract. In addition to committing to a direct signaling scheme, the sender also commits to a payment $p(i)$ for each action $i$. If the signaling scheme recommends action $i$, and the receiver follows the recommendation, she is then paid $p(i)$ by the sender (or pays the sender $-p(i)$ if $p(i) < 0$). If the receiver deviates from the scheme's recommendation, no payment is exchanged. % Note that the payment $p(i)$ depends only on the signaled action $i$, and not additionally on the state of nature; this guarantees that the payment is important so as not to reveal additional information beyond that revealed by the signaling scheme, and moreover is without loss.\footnote{A scheme which to depend on both the state of nature and the recommended action }
Since payments are exchanged only when the receiver follows the recommendation, nonnegative payments can be viewed as augmenting the ``persuasiveness'' of the signaling scheme. Negative payments are less natural, although their consideration will be technically instructive.

  We distinguish three classes of payment contracts: unrestricted (allowing arbitrary positive and negative payments), nonnegative, and budget balanced. For the latter, the sender's expected payment should be zero over states of nature and randomness in the signaling scheme, assuming the receiver(s) follow the recommendations of the scheme. Independent to our work, \citet{cheng_li} also considers adding payments to Bayesian persuasion. They analyze a special case of Bayesian persuasion with two states of nature, and examine how adding a payment contract influences the optimal policy in that scenario. Our approach diverges from this work by considering more general settings of Bayesian persuasion and various classes of payment contracts. 
  
\vspace{-3mm}
\paragraph{Duality as a Unifying Lens.}
%Duality: similarity between persuasion and MD, Myersonian perspective vs Borderian, canonical/simple rules in MD/AMD, lagrangify incentive constraints, 

Persuasion and auction design share striking parallels. Indeed, both are economic design problems in which the outputs --- recommendation(s) in the case of persuasion and allocations of goods in the case of an auction --- are subjected to incentive constraints which at the surface appear quite similar in the two settings. This is made explicit   by  \citet{DX-16}, who draw an analogy between persuasion and single-item auctions: actions are analogous to bidders, and recommending an action is analogous to allocating the item. Through this  analogy, they were able to leverage techniques from auction theory --- in particular Border's theorem~\citep{B-91} --- to characterize and compute optimal signaling schemes for Bayesian persuasion when action payoffs are i.i.d. This analogy is imperfect, however, as illustrated by the impossibility result of \citet{DX-16} for independent non-identical action payoffs, contrasting the tractability of single-item auctions with independent bidders. %The different incentive constraint the different nature of incentive constraints in contrast to the tractability of single item auctions with independent non-identical items.

Despite being imperfect, however, this similarity is suggestive: if a Border's theorem based approach  of optimization of interim rules can be applied to persuasion, why not the ``virtual value'' approach of \citet{Myerson81} as well? Myerson's approach can be viewed through the more general lens of Lagrangian duality, in particular as a consequence of Lagrangifying the incentive constraints. The duality-based approach has been applied to much more general mechanism design settings, producing (often approximate) generalizations of Myerson's virtual-value characterization which have led to simple and approximately optimal mechanisms in a number of multi-parameter settings (e.g. \cite{CDW-16, EdenFFTW17a, EdenFFTW17b, BrustleCWZ17, FULLT18, LiuP18, CaiZ17,haghpanah2015reverse}, and classic results such as \cite{rochet1998ironing}).  It is therefore natural that we embark on the same  exploration for persuasion, as well as for models (such as ours) which combine the approaches of persuasion and mechanism design. As a best case scenario, we can hope for ``simple'' characterizations of optimal or near-optimal schemes, akin to those derived from duality in mechanism design. Particularly attractive are ``canonical'' characterizations which depend minimally on the details of the instance at hand.

\vspace{-3mm}
\paragraph{Our Results.}
In Section~\ref{sec:lagrange}, we apply Lagrangian duality to Bayesian persuasion, and derive some elementary properties of the primal/dual pair which enable our results to follow. Our first main result is in Section~\ref{sec:nopayment}, and concerns (single-receiver) Bayesian persuasion with a prior distribution which is symmetric across $n$ actions\footnote{We refer the reader to~\cite{dughmi2017algorithmic,DX-16} for some natural examples of persuasion with symmetric actions.}, and no payments are allowed. We show that a single dual variable  naturally interpolates between two extreme schemes: at one extreme ($\lambda$ equals zero) we get the (non-persuasive) scheme which always recommends the sender's ex-post preferred action, at the other extreme ($\lambda$ very large) we get the (persuasive) scheme which recommends the receiver's ex-post preferred action. Intermediate values of $\lambda$ yield schemes which point-wise optimize the sender payoff plus $n\lambda$ times the receiver's payoff. Moreover, there is a threshold $\lambda^*$ below which the induced scheme is non-persuasive, and above which the scheme is persuasive. This $\lambda^*$ induces the sender-optimal persuasive signaling scheme. This characterization is detail-free, in the sense that it reduces the prior distribution to a relative weighting of receiver to sender payoff, in particular the minimum necessary for persuasiveness. Furthermore, this optimal scheme is \emph{Pareto efficient} in a strong (ex-post) sense:  for every state $\theta$, no outcome Pareto dominates the one picked by the scheme.

In Section~\ref{sec:payment}, we use duality to characterize Bayesian persuasion schemes with payments. When arbitrary payments are allowed and the prior is symmetric, the optimal signaling scheme is canonical and does not depend on the prior: it always (i.e. in every state of nature) recommends the action that maximizes the sender utility plus $\frac{n}{n-1}$ times the receiver's utility. Payments accompanying this scheme are computed easily via a simple payment identity. Our main result in this section is a \emph{dichotomy} for Bayesian persuasion with a symmetric prior, \emph{but} non-negative payments: the optimal scheme is either the same as the aforementioned arbitrary-payment scheme (in the event that non-negative payments are needed), or else is the optimal no-payment signaling scheme. Finally, with only two actions and an arbitrary prior, we show that when arbitrary payments are allowed, the optimal scheme always recommends the action maximizing sender utility plus twice receiver utility. Again, all of our optimal schemes are ex-post Pareto efficient. We note that the strongest positive results of~\cite{DX-16} (exact polynomial time solvability) hold in the setting of i.i.d. actions. Our results therefore extend and simplify theirs, while lending further insight.

In Section~\ref{sec:binary}, we turn our attention to a multi-receiver private persuasion model with externalities. Again, we employ duality to analyze the optimal scheme in this setting. Our first main result shows that, when we allow budget balanced payments, there exists an optimal scheme which is ``simple'' in the following sense: It always recommends an action maximizing a weighted sum of sender utility and the receivers' marginal utility from following the recommendation of the scheme. The relative weighting is  determined by a single dual variable.  This characterization is interesting when contrasted with the optimal no-payment scheme, which is not as simple in general, and we show is sometimes strictly outperformed (in terms of sender expected utility) by the optimal budget-balanced scheme. Our second main result is a generalization of an algorithmic result of \cite{DH-17} to multi-receiver persuasion with \emph{positive externalities}: when no payments are allowed, and sender and receiver utility functions lie in some cone of set functions $\mathcal C$, we use duality to exhibit a polynomial time reduction from optimal signaling to the optimization problem for set functions in $\mathcal C$. Due to the lack of space, details of this last result are in Appendix~\ref{sec: computation plus externalities}.

%\paragraph{Further Related Work.}
%
%\anote{todo}

%Rad: I don't think a further related work is even necessary. All the indirect related works have already been cited in the intor.

%\subsection*{Additional Discussion of Related Work}

%\subsection*{Organization}

%%% Local Variables:
%%% mode: latex
%%% TeX-master: "../main"
%%% End:

\newcommand{\ppolytope}{\mathcal{P}}
\newcommand{\polytope}{\mathcal{P}}
\newcommand{\nopayments}{ \{ \mathbf{0}^n \} }
\newcommand{\dual}[2]{\lambda(#1,#2)}
\newcommand{\nstate}{\theta}
\newcommand{\scheme}[1]{\phi_\nstate(#1)}
\newcommand{\signalS}{\Sigma}
\newcommand{\rpay}[1]{r_{\nstate}(#1)}
\newcommand{\spay}[1]{s_{\nstate}(#1)}
\newcommand{\dist}{\mu_\nstate}
\newcommand{\gpayment}[1]{p(#1)}
\newcommand{\sipay}[2]{\xi^{#1}_{#2}}
\newcommand{\ripay}[2]{\rho^{#1}_{#2}}
\newcommand{\siidpay}[1]{\xi_{#1}}
\newcommand{\riidpay}[1]{\rho_{#1}}
\newcommand{\siidpayv}{\pmb{\xi}}
\newcommand{\riidpayv}{\pmb{\rho}}
\newcommand{\disti}[1]{\mu_{\nstate_#1}}
\newcommand{\schemev}{\phi_\nstate}
\newcommand{\signv}{\mathcal{M}}
\newcommand{\signijk}[3]{M^{#1}_{#2#3}}
\newcommand{\signi}[1]{M^{#1}}
\newcommand{\spolytope}{\mathcal{P}_s}
\newcommand{\lagrange}{\lambda}
\newcommand{\lagrangeF}{\mathcal{L}}
\newcommand{\dualv}{\lambda}
\newcommand{\dualr}[1]{r^\dualv_{\nstate}(#1)}
\newcommand{\distn}{\mu}
\section{Preliminaries}
\label{sec:prelim}
%\subsection{Notations and basics of Bayesian persuasion}
%\label{sec:persuasion-basic}
\vspace{-3mm}
\paragraph{Bayesian persuasion.} Bayesian persuasion is a game between a  sender, also termed as the \emph{principal}, and a receiver, also termed as the \emph{agent}. There is also a set of possible \emph{states of nature} $\SSpace$. The true state of nature $\nstate\in \SSpace$ is drawn from a prior distribution $\dist$, known to both the sender and the receiver. The receiver has a set of possible \emph{actions} $[n] = \{ 1,\dots, n \}$ to pick. Based on the action $i$ picked by the receiver and the true state of nature $\nstate$,  the sender and the receiver gain \emph{payoffs}\footnote{ We use the words ``payoff" and ``reward" interchangeably in this paper.}, denoted by $\spay{i}$ and $\rpay{i}$ respectively. In the Bayesian persuasion game, the sender commits to a \emph{signaling scheme} $\phi$, where in general $\phi$ is a mapping from $\SSpace$ to distributions over possible \emph{signals},~$\signalS$. Then the sender observes $\nstate \sim \dist$, and sends a signal $\sigma\in\signalS$ to the receiver, where $\sigma\sim \phi_\theta$ for the observed $\theta$. Given signal $\sigma$, the receiver updates her belief about the state of nature $\nstate$, and selects an action $i_r$ that maximizes her expected payoff under this posterior distribution. In this paper, without loss of generality, and by applying the revelation principle~\citep{KG-11}, we focus on signaling schemes for which $\signalS=[n]$. We use the notation $\scheme{i}$ to denote the probability that the sender recommends action $i$ conditioned on the state of nature being $\nstate$. A signaling scheme is said to \emph{implement} $\schemev$, if it samples signal $i\sim \schemev$ given the state of nature $\nstate$.

\vspace{-5mm}
\paragraph{Persuasive signaling and optimal Bayesian persuasion.} A signaling scheme is \emph{persuasive} if the receiver is best off following the sender's recommendation, i.e. following the sender's recommendation maximizes the receiver's expected payoff under the receiver's posterior belief about the state of the nature (conditioned on the received signal).  In the \emph{optimal Bayesian persuasion problem}, we wish to find,  over all persuasive schemes, the one that maximizes the sender's expected payoff. % when the receiver follows her recommendation.

\vspace{-5mm}
\paragraph{Action types.}  As typical in information structure design, we frequently think of each action as having a ``type'' depending on the state of nature.\footnote{ We refer the reader to \cite{DX-16} for a list examples of Bayesian persuasions and how types are defined in those.} That is, $\nstate=[\nstate_1,\ldots,\nstate_n]$ is a vector in $[m]^n$, for some parameter $m$. Action $i$ has type $\nstate_i$, which completely determines the sender and receiver payoffs should action $i$ be selected, independent of $\nstate_{-i}$. More clearly, there exist $m$ pairs of payoffs $(\sipay{i}{1},\ripay{i}{1}),\ldots,(\sipay{i}{m},\ripay{i}{m})$ such that when the receiver selects action $i$ with $\nstate_i = j$, $\sipay{i}{j}$ is the payoff to the sender and $\ripay{i}{j}$ is the payoff to the receiver. Note that $\spay{i}=\sipay{i}{\nstate_i}$ and $\rpay{i}=\ripay{i}{\nstate_i}$, and for universality we stick to this notation (see Appendix~\ref{appendix:reduced form}). Distributions in this setting may be \emph{independent}, if $\dist = \underset{i}{\times}~\disti{i}$ is a product distribution, where each $\disti{i}$ is a distribution over $[m]$. Distributions in this setting may also be \emph{symmetric}, if $\dist$ is invariant under all permutations. Distributions that are both independent and symmetric are \emph{i.i.d}. 

\vspace{-5mm}
\paragraph{Bayesian persuasion with payments.} In this paper, we introduce a natural model of payments into the Bayesian persuasion problem. In addition to recommending an action $i$, the sender is allowed to ``incentivize'' the receiver to take the recommendation with an additional payment $\gpayment{i}$. Mathematically, it also makes sense to consider when payments are allowed to be negative (which corresponds to the sender ``charging'' the receiver $-\gpayment{i}$ in order to follow the recommendation). This certainly could be relevant for practice (e.g. if the sender/receiver can commit to contracts), but the non-negative payment model is clearly more natural. We study four different payment models: \emph{zero payments} (i.e. classic Bayesian persuasion), \emph{non-negative payments} (i.e. when the sender cannot charge the receiver), \emph{budget-balanced payments} (i.e. the payment of the sender is zero in expectation) and \emph{general payments} (i.e. payments are arbitrary real numbers). To unify the notation throughout the paper, we use ${\polytope}$ to denote the feasible set of payments, which plays the role of a different polytope for each relevant payment model (note that in all four models $\polytope$ is closed and convex). A signaling scheme is said to \emph{implement} $(\schemev,p)$ if given the observed state of nature $\nstate$, it samples $i\sim \schemev$ and pays the receiver a (randomized) payment $\gpayment{i}$. Throughout this paper, it will be much more convenient to focus on specifying the \emph{expected payments} for following the recommendation of action $i$, $P(i) \triangleq \sum_{\nstate\in\SSpace} \dist \scheme{i}\gpayment{i}$ for any scheme. We similarly define the implementability of $(\schemev,P)$.\footnote{ It is also easy to see how to implement payments $\gpayment{i}$ from $P(i)$ and an implementation of $\schemev$: for a given state $\nstate\in \SSpace$, sample a payment $\tfrac{P(i)}{\scheme{i}}$ whenever signal $i\sim \schemev$ is recommended.}

\paragraph{Payment identity and optimal payments.} We conclude with an observation about Bayesian persuasion with payments. Similar to auction design, there is a ``payment identity'' capturing which payments will make an implementable signaling scheme persuasive. In contrast to auctions, however, it is easy to see that \emph{every} signaling scheme can be made persuasive with sufficiently high payments. 
\vspace{-2mm}
\begin{observation}\label{obs:payments}
Let $\dist$ be any distribution over states of nature, and $\schemev$ be any signaling scheme (not necessarily persuasive). Then there exist thresholds $T_1,\ldots, T_n$ such that $(\schemev,p)$ is persuasive if and only if $P(i) \geq T_i$ for all $i$. 
\end{observation}
\begin{proof}
Let $X_{\phi}(i,j)$ be the receiver's expected utility by taking action $j$ when the scheme $\schemev$ recommends action $i$. Define $T_i = \max_{j \neq i} \{X_{\phi}(i,j) - X_{\phi}(i,i)\}$. Therefore, $X_{\phi}(i,i) + T_i \geq X_{\phi}(i,j)$ for all $j$, and any scheme that pays $P(i) \geq T_i$ is certainly persuasive when recommending $i$. Moreover, if $P(i) < T_i$, then there exists some $j$ s.t. $X_{\phi}(i,i) + P(i) < X_{\phi}(i,j)$, and the scheme is not persuasive. 
\end{proof}

\vspace{-3mm}
\begin{definition} We refer to the $T_1,\ldots, T_n$ guaranteed in Observation~\ref{obs:payments} as the \emph{optimal payments} for $\phi$, and $\max\{0,T_1\},\ldots, \max\{0,T_n\}$ as the \emph{optimal non-negative payments} for $\phi$.
\end{definition}

\vspace{-6mm}
\paragraph{Bayesian persuasion with multiple agents, binary actions and externalities.} In this paper, we formalize a model with externalities, naturally extending the model introduced recently in~\cite{AB-16} and \cite{DH-17} for private persuasion with no externalities.  Consider the general setup of Bayesian persuasion with one sender and  $\agents$ receivers, where each receiver can take either action $0$ or action $1$. Using the revelation principle for Bayesian persuasion, we can restrict our attention to direct signaling schemes, where a direct (randomized) scheme can be thought of as a mapping from $\SSpace$ to distributions over subsets of $[\agents]$, indicating to which receivers action $1$ is recommended. We denote such a scheme by $\BScheme{S}$ for every $\theta \in \SSpace, S\subseteq [\agents]$. Given the subset $S$ of receivers taking action $1$, let $\util{i}{S}$ be the payoff of receiver $i$ and $\Ffun{S}$ be sender's payoff. We further assume that $\Ffun{S}$ is a monotone set function. We also allow payments of the form $\{\pay{i}{a}\}_{i\in[\agents], a\in\{0,1\}}$, where $\pay{i}{a}$ is the payment agent $i$ receives by following action $a$, if it is recommended. We occasionally use the notation $\Gfun{i}{S}\triangleq \util{i}{S}-\util{i}{S\setminus\{i\}}$, which are called \emph{marginal utilities} in this paper. For our computational results, we further assume \emph{positive externalities}, i.e. for every agent $i$, $\Gfun{i}{S}$ can only increase if agent $j\neq i$ switches from action $0$ to $1$.

\vspace{-2mm}
\section{Lagrangian Duality and Bayesian persuasion}
\label{sec:lagrange}
In this section, we create a unified analysis toolbox for various Bayesian persuasion problems through the lens of LP duality. Specifically, we use Lagrangian duals to reveal the structures of the optimal signaling, \'a la successful instances of a similar technique in Bayesian mechanism design~\citep{CDW-16}. As a prerequisite, we heavily use linear programming techniques introduced in~\cite{KG-11}, and further improved in \cite{DX-16}, and build a bedrock for the analyses in future sections.

\vspace{-3mm}
\subsection{LP for general Bayesian persuasion with payments} Linear programming formulations of Bayesian persuasion without payments have been introduced in numerous prior works (e.g.~\cite{DX-16}). We now modify the program slightly to capture payments, by observing that whenever action $i$ is recommended the receiver gets additional payoff $\gpayment{i}$ and the sender loses a payoff  $\gpayment{i}$. The following program solves the sender's optimization problem. 
\begin{align*}
\max~~&\sum_{\nstate \in \SSpace} \sum_{i \in [n] } \dist \scheme{i} \left( \spay{i}- \gpayment{i} \right) &\\
&\sum_{\nstate \in \SSpace} \dist\scheme{i} \left( \rpay{i} +\gpayment{i} \right) \geq \sum_{\nstate \in \SSpace} \dist \scheme{i} \rpay{j}, & \forall  i,j \neq i \in [n] \\
&\schemev\in \Delta_n,~\forall \nstate\in\SSpace,~~~\text{and}~~~{p} \in \polytope
\end{align*}

The first set of constraints, also called \emph{persuasiveness constraints}, are similar to \emph{Incentive Compatibility (IC)} constraints from auction design, and  ensure that the receiver is best off (in expectation) by following the recommendation and paying the payment. The rest ensure that the scheme is in fact a valid distribution over recommendations, and the payments are feasible. While not yet a linear program, the above program can be made linear by a simple change of variables to expected payments $P(i)$:\footnote{ Recall that $P(i) = \sum_{\nstate \in \SSpace} \dist \phi_{\nstate}(i) \gpayment{i}$.}
\begin{talign*}
\max~~&\sum_{\nstate \in \SSpace} \sum_{i \in [n] } \dist \scheme{i}\spay{i}- \sum_{i\in[n]} P(i)  &\customlabel{eq:Lp-general-payments}{\textit{(LP-General-with-Payments)}}\\
&P(i)+\sum_{\nstate \in \SSpace} \dist\scheme{i} \rpay{i}  \geq \sum_{\nstate \in \SSpace} \dist \scheme{i} \rpay{j}, & \forall  i,j \neq i \in [n] \\
&\schemev\in \Delta_n,~\forall \nstate\in\SSpace,~~~\text{and}~~~P \in \polytope
\end{talign*}
where we abuse the notation by using $\polytope$ to denote the feasible polytope of average payments $P$.
\vspace{-3mm}
\subsection{The partial Lagrangian dual}
\label{sec:partial lagrangian}

 One of the main tools we use in this paper is taking the \emph{partial Lagrangians} of the linear programs of the Bayesian persuasion problem. Basically, similar to \cite{CDW-16}, we do not take a ``complete dual'' of the LP formulation, instead ``Lagrangifying'' only the persuasiveness constraints, and leave all feasibility constraints in the primal.\footnote{ There are significant conceptual differences between ``incentive constraints'' in auction design and ``persuasiveness constraints''.
Some of these differences were briefly discussed in Section~\ref{sec:intro}; we will highlight them further in future technical sections.}  To take the partial Lagrangians, we apply the method of Lagrangian multipliers by (\emph{a}) introducing dual variables $\dual{i}{j}$ for every pair of actions $i,j \in [n]$ with $i \neq j$, (\emph{b}) multiplying the persuasiveness constraints with duals $\dual{i}{j}$,  and (\emph{c}) moving the persuasiveness constraints to the objective. By rearranging the terms we have the following observation.
\vspace{-2mm}
\begin{observation}\label{obs:lagrangians} Assigning dual variables $\dual{i}{j}$ to the persuasiveness constraint guaranteeing that the receiver prefers to take action $i$ over action $j$ when $i$ is recommended gives the following partial Lagrangian.
\begin{align}
\lagrangeF_{\dualv} (\phi,P) = \sum_{\nstate\in\SSpace,i \in [n]}\dist\scheme{i} &\left( \spay{i} + \rpay{i} \sum_{j \neq i} \dual{i}{j} - \sum_{j \neq i} \dual{i}{j} \rpay{j} \right) + \sum_{i\in[n]} P(i) \left( \sum_{j \neq i} \dual{i}{j} - 1 \right). \label{eq:Lagrange equation:main}
\end{align}
\end{observation}

\noindent We conclude by reminding the reader that the optimal signaling schemes and payments are solutions to the following min-max programs by applying strong duality ($\mathcal{D}$ denotes the appropriate dual feasible polytope):
\begin{align}
\max_{\forall \theta:\schemev\in\Delta_n, P\in\polytope}\left(\min_{\dualv\in \mathcal{D}}\lagrangeF_{\dualv} (\phi,P)\right)=\min_{\dualv\in \mathcal{D}}\left(\max_{\forall \theta:\schemev\in\Delta_n, P\in\polytope}\lagrangeF_{\dualv} (\phi,P)\right).\label{eq:strongduality}
\end{align}

\begin{definition}[\emph{Dual-adjusted receiver payoff}] For any assignment of dual variables $\dualv\in\mathcal{D}$, define $\dualr{i} \triangleq \rpay{i} \sum_{j \neq i} \dual{i}{j} - \sum_{j \neq i} \dual{i}{j} \rpay{j}$.
\end{definition}
We conclude by Proposition~\ref{prop:duality}, which we repeatedly use in Sections~\ref{sec:nopayment},~\ref{sec:payment} and~\ref{sec:binary} to extract various properties of the corresponding optimal policy. 

\begin{proposition}\label{prop:duality}[\emph{Strong Duality for Bayesian Persuasion}] There exist dual variables $\dualv(.,.)$ such that the optimal signaling scheme, for every state of nature $\theta$, recommends the action maximizing the dual-adjusted receiver payoff, i.e. $\spay{i} + \dualr{i}$. Moreover, if $\dual{i}{j} > 0$, then when action $i$ is recommended, the receiver is indifferent between following the recommendation and taking action $j$ instead (Complementary Slackness).\footnote{ We will not actually make use of complementary slackness in this paper, but include it here for completeness.} 
\end{proposition}
\begin{proof}

This is is an immediate consequence of strong duality, plus some observations about the structure of the Bayesian Persuasion LP. Observe first that indeed $\lagrangeF_{\dualv}(\phi,P) =\sum_{\nstate\in\SSpace,i \in [n]}\dist\scheme{i}\cdot (\spay{i} + \dualr{i}) +  \sum_{i\in[n]} P(i) ( \sum_{j \neq i} \dual{i}{j}- 1 )$. Now, consider the RHS of Equation~\ref{eq:strongduality}. For given dual variables $\dualv$, observe that the right-most term (depending on $P$) does not depend on $\phi$ at all. Observe that the remaining term ($\sum_{\nstate\in\SSpace,i \in [n]}\dist\scheme{i}\cdot (\spay{i} + \dualr{i})$) is trivial to optimize over all $\phi$ such that $\phi_{\theta} \in \Delta_n$ for all $\theta \in \Theta$: simply recommend the action maximizing $\spay{i} + \dualr{i}$ for all $\theta \in \Theta$. Therefore, when $\dualv$ are the optimal dual variables, the optimal signaling scheme must recommend the action maximizing $\spay{i} + \dualr{i}$ on all states of nature. Complementary slackness follows immediately from strong duality.
\end{proof}

\subsection{Exploiting symmetries}
When viewing actions by their types (recall this means that we view states of nature as a profile $[\nstate_1,\ldots, \nstate_n]$, with each $\nstate_i \in [m]$), some of our results consider symmetric settings, where $\mu_{\nstate_1,\ldots, \nstate_n} = \mu_{\nstate_{\pi(1)},\ldots,\nstate_{\pi(n)}}$ for any permutation $\pi:[n]\rightarrow [n]$. It is well-known that symmetric LPs admit symmetric solutions (e.g.~\cite{HB-09}), and this fact has indeed been exploited in prior work on signaling and auctions~\citep{DX-16,CDW-12}. The proof of this is straight-forward, but we sketch it below (and provide a full proof in Appendix~\ref{sec:missing-proofs}). First, we state clearly what we mean by symmetric solutions.

\begin{definition}We say that a signaling scheme $(\phi, P)$ is symmetric if $P(i) = P(j)$ for all $i,j$ and for all permutations $\pi$, we have $\phi_{[\nstate_1,\ldots, \nstate_n]}(i) = \phi_{{[\nstate_{\pi(1)},\ldots,\nstate_{\pi(n)}]}}(\pi^{-1}(i))$ for all $i, \nstate$. We say that a dual solution $\dualv$ is symmetric if $\dualv(i,j) = \dualv(k,\ell)$ for all $i,j, k,\ell$. 
\end{definition}

\begin{proposition}\label{prop:symmetry} Let $\dist$ be a symmetric instance of Bayesian Persuasion, with any of the four referenced constraints on payments (none, non-negative, budget-balanced, or arbitrary). Then there exists an optimal symmetric primal and an optimal symmetric dual for $\dist$.
\end{proposition}

\begin{proof}[Proof sketch]
The proof is essentially a formalization of the following intuition. First, if a scheme $\phi$ is optimal and persuasive, and $\dist$ is symmetric, then we can first ``relabel'' the actions according to any permutation $\pi$ before implementing $\phi$, and it will still be optimal and persuasive. Second, any distribution over optimal, persuasive schemes is still optimal and persuasive. And third, if we consider the scheme that randomly samples a permutation $\pi$ with which to relabel the actions before implementing $\phi$, then this scheme is symmetric. A full proof appears in Appendix~\ref{sec:missing-proofs}.
\end{proof}

With Proposition~\ref{prop:symmetry} in hand, we can now draw further conclusions regarding the format of the Lagrangian function $\lagrangeF$ in the special case that $\dist$ is symmetric:
\begin{corollary}\label{cor:symmetryLagrange} When $\dist$ is symmetric, there exists a constant $\lagrange$ such that the optimal scheme, on every state of nature $\nstate$, selects an action $i$ maximizing $\spay{i} + n\lagrange \rpay{i}$. Moreover, for the same $\lagrange$ and some constant $C$, the Lagrangian takes the following form:
\[ \lagrangeF_{\lagrange}(\phi,P) = \sum_{\nstate\in\SSpace,i \in [n]}\dist\scheme{i} \left( \spay{i} + n \lagrange \rpay{i}  \right) + \sum_{i\in[n]} P(i) \left( (n-1) \lagrange - 1 \right) - \lagrange C. \]
\end{corollary}

\begin{proof}
Consider an optimal and symmetric dual solution, which is guaranteed to exist by Proposition~\ref{prop:symmetry}, in which $\dualv(i,j) = \lagrange$ for all $i \neq j$. Then we get the following simplified form for $\lagrangeF$:
\begin{talign*}
\lagrangeF_{\dualv}& (\phi,P) = \sum_{\nstate\in\SSpace,i \in [n]}\dist\scheme{i} \left( \spay{i} + \rpay{i} \sum_{j \neq i} \lagrange - \sum_{j \neq i} \lagrange \rpay{j} \right) + \sum_{i\in[n]} P(i) \left( \sum_{j \neq i} \lagrange - 1 \right) \\
&= \sum_{\nstate\in\SSpace,i \in [n]}\dist\scheme{i} \left( \spay{i} + \rpay{i} (n-1) \lagrange - \lagrange \sum_{j \neq i}  \rpay{j} \right) + \sum_{i\in[n]} P(i) \left( (n-1) \lagrange - 1 \right) \\
&= \sum_{\nstate\in\SSpace,i \in [n]}\dist\scheme{i} \left( \spay{i} + n \lagrange \rpay{i}  \right) + \sum_{i\in[n]} P(i) \left( (n-1) \lagrange - 1 \right) - \sum_{\nstate\in\SSpace,i \in [n]}\dist\scheme{i} \lagrange \sum_{j \in [n]}  \rpay{j} \\
&= \sum_{\nstate\in\SSpace,i \in [n]}\dist\scheme{i} \left( \spay{i} + n \lagrange \rpay{i}  \right) + \sum_{i\in[n]} P(i) \left( (n-1) \lagrange - 1 \right) - \lagrange C,
\end{talign*}
Where we have defined $C = \sum_{\nstate \in \SSpace} \dist \phi_{\nstate} \sum_{j \in [n]} \rpay{j}$. It is now easy to see that the signaling scheme maximizing $\lagrangeF_{\dualv}$ necessarily on every state of nature $\theta$ recommends an action maximizing $\spay{i} + n\lagrange \rpay{i}$. 
\end{proof}

\section{Illustrative Examples}\label{sec:examples}

Here, we work through two illustrative examples to highlight the role of different payment models in Bayesian persuasion. Our examples shed some insights on the payment choices. We formalize these insights in the later sections. 

%\rnote{I think this section needs to be cleaned heavily, and the language needs to be a little bit more formal.}

\subsection{A tempting but false argument}\label{sec:tempting}

Initially, it may seem that the optimal scheme with arbitrary payments must recommend an action in $\argmax_i\{\spay{i} + \rpay{i}\}$ for all $\nstate \in \SSpace$ (and pay the optimal payments). In fact, the following argument seems to confirm this intuition.
Consider an optimal persuasive signaling scheme $\phi$ and assume (towards a contradiction) that there is some state of nature $\nstate$ where $\phi$ does \emph{not} recommend an element of $\arg \max_i\{\spay{i} + \rpay{i}\}$. Let $j$ denote the action recommended, and $i$ denote an element in the argmax. Modify $\phi$ by moving all the probability mass from $\phi_{\nstate}(j)$ to $\phi_{\nstate}(i)$, i.e. recommend action $i$ in state $\nstate$, and increase $P(i)$ by $\dist \cdot \phi_{\nstate}(j) \cdot (\rpay{j}-\rpay{i})$. Let $\phi'$ be the new signaling scheme. The following two observations are immediate. First, the sender's payoff under $\phi'$ is strictly larger than the sender's payoff under $\phi$. This is because the sender pays an additional $\dist \phi_{\nstate}(j) (\rpay{j} - \rpay{i})$, and gets additional payoff $\dist \phi_{\nstate}(j) (\spay{i}-\spay{j})$; the total change in sender payoff is $\dist \phi_{\nstate}(j) \left(\spay{i} + \rpay{i} - \spay{j} - \rpay{j}\right) > 0$, by the hypothesis about action $i$. Second, the receiver's payoff under $\phi'$ is equal to the receiver's payoff under $\phi$. The receiver gets an additional $\dist \phi_{\nstate}(j) (\rpay{j}-\rpay{i})$ in payment, and gets additional payoff $\dist \phi_{\nstate}(j)(\rpay{i} - \rpay{j})$. Note that one of these terms might be negative.

Therefore, the sender's payoff in $\phi'$ is strictly improved, while the receiver is just as happy, i.e., it seems like we have reached a contradiction to the optimality of $\phi$. The catch is that $\phi'$ \emph{is not necessarily persuasive!} %It is possible to have a non-persuasive scheme $\phi'$ that gives the receiver higher expected payoff than a persuasive scheme $\phi$. 
Since payments in Bayesian persuasion are quite different than payments in (say) auctions, we find it educational to present here a concrete example that confirms that the scheme described is suboptimal.\footnote{The proofs of our theorems of course provide confirmation themselves, but since the intuition is especially jarring for those unfamiliar with Bayesian persuasion, the example should be instructive.}

%The first example we present below works through a simple example to confirm that it is indeed better to use the optimal scheme proposed by our theorems, and not the scheme suggested by this tempting logic.

Consider the following. There are two actions, $A$ and $B$. Each action gives the receiver payoff $0$ or $1$, uniform and iid. Identically for the sender payoff. In other words, there are four possible action types $(0,0), (0,1), (1,0)$ and $(1,1)$, encoding the receiver and sender payoffs respectively. 
We will compare the following two schemes: (1) Recommend a uniformly random action in $\arg \max_i \{ \spay{i} + \rpay{i} \}$, and pay the optimal payments (the scheme suggested above), (2) Recommend a uniformly random action that maximizes $\spay{i} + 2\rpay{i}$, and pay the optimal payments (the optimal scheme, as suggested by Proposition~\ref{prop:general arbitrary binary actions}).

Towards analyzing the first scheme, observe that conditioned on action $A$ having type $(1,1)$, it is recommended with probability $7/8$: it is always recommended unless $B$ has type $(1,1)$, in which case it is recommended with probability $1/2$. Conditioned on action $A$ having type $(0,1)$ or $(1,0)$, it is recommended with probability $1/2$, and conditioned on having type $(0,0)$ it is recommended with probability $1/8$. The expected receiver payoff for taking the recommendation, conditioned on action $A$ being recommended, is $(7/8+1/2+0+0)/2 = 11/16$. The expected receiver payoff for taking action $B$ when $A$ was recommended is $5/16$. Therefore, the optimal payments are $-6/16$ for each action. The sender's expected payoff when the recommendation is followed is also $11/16$, so the total sender utility (accounting for payments) is $17/16$. 

Let's analyze the second scheme. The difference between the two schemes is that the second one tie-breaks in favor of $(0,1)$ over $(1,0)$. So, conditioned on action $A$ having type $(1,1)$, it's still recommended with probability $7/8$. Conditioned on having type $(0,1)$, it's recommended with probability $5/8$, type $(1,0)$ is recommended with probability $3/8$, and type $(0,0)$ is recommended with probability $1/8$. The expected receiver payoff, conditioned on action $A$ being recommended, is now $(7/8 + 5/8+0+0)/2 = 3/4$. The expected payoff for taking action $B$ when $A$ was recommended is then $1/4$. This means that the optimal payments are $-1/2$ for each action. The sender's expected payoff when the recommendation is followed is $(7/8+0 + 3/8+0)/2 = 5/8$. The sender's total utility is $9/8$, indeed better than $17/16$.
%Again, the purpose of this is just to give the reader some experience with a concrete example, as payments in Bayesian persuasion are quite different than payments in (say) auctions. 

%\subsection{Example two: contrasting payment schemes}
\subsection{Contrasting payment schemes}
\label{sec:examples-contrast}

The following example demonstrates how negative payments can be quite counterintuitive in Bayesian persuasion, and how various constraints on payments differ. 
Consider an instance with one sender and one receiver that has two possible actions 0 and 1.  Suppose that $\SSpace=\{0,1\}$, that the two states are equiprobable, and that the payoff functions of the sender and the receiver satisfy: $s_\theta(0)=-r_\theta(0)={1-\theta}$ and $s_\theta(1)=-r_\theta(1)={1+\theta}$. The payoffs of the  sender and the receiver always sum up to zero, so  this is a zero-sum persuasion game. 

The optimal signaling without payments is to not reveal any information to the receiver. To see this, observe that the receiver must get payoff at least $-1/2$, as they can guarantee this by ignoring the sender's signal and taking action $0$. Thus, the sender cannot get utility strictly bigger than $1/2$, which they can guarantee by recommending action $0$. 
This scheme is in fact optimal even when the sender is allowed to use non-negative payments. To see this, observe that the receiver's utility is again at least $-1/2$: they can always take action $0$ regardless of the recommendation, and be paid a non-negative amount. Again, the sender cannot get utility strictly bigger than $1/2$. 

Next, consider the following scheme with negative payments, as suggested by Proposition~\ref{prop:general arbitrary binary actions}: always recommend action $0$ and charge $1$. This scheme is persuasive, because the receiver's expected payoff for taking action $1$ is $-3/2$. This gives the sender expected utility equal to $3/2$.

Finally, consider the following budget balanced scheme, recommended by Proposition~\ref{thm:budget-balance}. When $\nstate = 1$, recommend action $1$ with probability $q$. When $\nstate = 0$, recommend action $0$ deterministically. The receiver payoff for taking action $1$, conditioned on it being recommended, is $-2$. The receiver payoff for taking action $0$, conditioned on $1$ being recommended, is $0$. Therefore, the sender must pay $2$ to make this persuasive, for an expected total payment of $2 \cdot q \cdot 1/2 = q$. Conditioned on action $0$ being recommended, the receiver payoff for taking action $0$ is $(-0.5)/(1-q/2)$, while her payoff for taking action $1$ is $(-0.5 -(1-q))/(1-q/2)$. The sender can charge $(1-q)/(1-q/2)$ and remain persuasive, for a total expected charge of $1-q$. The total sender payoff (not counting payments) is then $2 \cdot q \cdot 1/2 + 0 \cdot (1-q) \cdot 1/2 + 1 \cdot 1/2 = q+1/2$. Accounting for payments, the sender pays $q$, and charges $(1-q)$, for an additional payoff of $1-2q$. The sender's total utility is $3/2 - q$. If $q = 0$, we recover the optimal scheme with possibly negative payments. If $q = 1/2$, then the total expected payment is $0$, and the scheme is budget balanced. In this case, the expected sender utility is $1$, strictly bigger than the optimal expected utility with non-negative payments.
 
The above example highlights the following concepts. First, the purpose of negative payments is to get extra payoff out of ``strictly persuasive'' schemes. Second, depending on the context, negative payments might not always be well motivated. For example, if the setup is such that a sender is making a non-binding recommendation to a receiver, then in the above examples the receiver would simply choose not to engage in the recommendation and always take action $0$. On the other hand, if the sender acts as a ``gatekeeper'' to the actions, e.g. because the actions are which event to attend that the sender is hosting or which fund to invest in that the sender manages, then negative payments are well motivated. Third, all three payment methods considered in this paper are distinct. 

\vspace{-3mm}
\section{Symmetric Actions and No Payments}
\label{sec:nopayment}
Here, we consider the standard symmetric setting without payments. We show how to derive structure on the optimal scheme by making use of duality. In any symmetric instance of Bayesian persuasion, there exists an optimal dual solution $\dualv^{*}(i,j)$ such that $\dualv^{*}(i,j)=\lagrange^*\geq 0$ for all $i\neq j$ by applying Proposition~\ref{prop:symmetry}. Now, by plugging directly into Corollary~\ref{cor:symmetryLagrange} and observing that $P(i) = 0$ for all $i$, we get that the Lagrangian for the optimal dual takes the following form:

\begin{equation}\label{eq: no payment iid}
\lagrangeF_{\lagrange}(\phi) = \sum_{\nstate\in\SSpace,i \in [n]}\dist\scheme{i} \left( \spay{i} + n \lagrange \rpay{i}  \right) - \lagrange C.
\end{equation}

%Notice that a persuasive scheme satisfies $\sum_{k \in [m]} \rho_k (x_k -  y_k) \geq 0$.
\noindent We immediately conclude that the optimal scheme recommends, for every state $\nstate$, an action maximizing $\spay{i} + n\lagrange^*\rpay{i}$. With a little more work, we can conclude something stronger about the exact value of $\lagrange^*$. 

\begin{definition}[$\lagrange$-scaled welfare maximizer] For a given multiplier $\lagrange$, define $\phi^{\lagrange}$ to be the scheme with $\phi^{\lagrange}_{\nstate}(i) = 0$ if $i \notin \arg\max_{j}\{\spay{j} + n\lagrange\rpay{j}\}$, and $\phi^{\lagrange}_{\nstate}(i) = 1/|\{\arg \max_j\{\spay{j}+n\lagrange\rpay{j}\}|$ if $i \in \arg\max_{j}\{\spay{j} + n\lagrange\rpay{j}\}$. In other words, $\phi^{\lagrange}$ recommends a uniformly random action in $\arg\max_{j}\{\spay{j} + n \lagrange\rpay{j}\}$.  
\end{definition}

\begin{proposition}\label{prop:symmetricnomoney} Let $\lagrange^*$ be the smallest $\lagrange \geq 0$ such that the scheme $\phi^{\lagrange}$ is persuasive for $\dist$. Then $\phi^{\lagrange^*}$ is the optimal scheme for $\dist$.
\end{proposition}
\begin{proof}
The proof will follow from the following. We claim that the persuasiveness of $\phi^{\lagrange}$ is monotone increasing in $\lagrange$ (larger $\lagrange$ is more persuasive), while the sender payoff is monotone \emph{decreasing} in $\lagrange$. Together this immediately concludes that the optimal persuasive scheme is $\phi^{\lagrange^*}$. We'll first need a technical lemma.

\begin{lemma}\label{lem:symmetric} Let $\phi^1$, $\phi^2$ be symmetric schemes and let $\dist$ be symmetric. Let the expected receiver payoff for accepting recommendation $\phi^1$ for $\dist$ be at least as large as the expected receiver payoff for accepting recommendation $\phi^2$ for $\dist$. Then if $\phi^2$ is persuasive, so is $\phi^1$.\footnote{ Note that this does \emph{not} hold generally, and absolutely requires the symmetry assumptions.}
\end{lemma}
\begin{proof}
First, observe that by symmetry, the expected payoff for any action conditioned on that action being recommended is the same. Moreover, the expected payoff for any action $i$ conditioned on action $j \neq i$ being recommended is also the same (for all $i \neq j$). Therefore, the scheme is persuasive as long as for all actions, the expected payoff conditioned on being recommended exceeds the expected payoff conditioned on not being recommended. This holds if and only if the expected payoff conditioned on being recommended exceeds the unconditional expected payoff. As the unconditional expected payoff for an action is independent of the payoff scheme (denote it by $C$), we conclude that a signaling scheme is persuasive if and only if the expected receiver payoff for following the recommendation exceeds $C$. The lemma immediately follows.
\end{proof}

\begin{lemma}\label{lemma: monotone symmetric}
If $\phi^{\lagrange}$ is persuasive, then $\phi^{\lagrange+\delta}$ is persuasive, for all $\delta \geq 0$.
\end{lemma}

\begin{proof}
First, observe that $\phi^{\lagrange}$ is symmetric for all $\lagrange$. Further observe that the receiver's expected payoff for following the recommendation is monotone in $\lagrange$: On every state of nature $\nstate$, the recommended action maximizes $\spay{i} + n \lagrange \rpay{i}$. The proof now immediately follows by Lemma~\ref{lem:symmetric}. 
\end{proof}

\begin{lemma}
The sender's expected payoff when the receiver follows $\phi^{\lagrange}$ is monotone non-increasing in $\lagrange$.
\end{lemma}
\begin{proof}
Simply observe that on every state of nature $\nstate$, the recommended action maximizes $\spay{i}+n\lagrange\rpay{i}$. As $\lagrange$ increases, the sender payoff for the recommended action decreases.
\end{proof}

The proof of Proposition~\ref{prop:symmetricnomoney} now immediately follows.
\end{proof}

\section{Optimal Single Agent Signaling with Payments}
\label{sec:payment}
In this section we study the single receiver Bayesian persuasion game with payments. While we consider our ``main results'' to be the case where payments are constrained to be non-negative, it's instructive to study general (positive or negative) payments. We characterize the optimal signaling scheme with payments in the general setting, drawing similar conclusions to~\cite{CDW-16} for optimal auctions. This is not a main result, but may be of independent interest. An easy corollary of this characterization, however, immediately allows us to claim something interesting in the case of two actions. The optimal scheme, for all states $\theta$, recommends the action $i$ maximizing $\spay{i} + 2\rpay{i}$ (paying the optimal payments), and this holds for any distribution. 

A deeper application of this characterization lets us characterize the optimal scheme for a single receiver with $n$ symmetric actions. In this setting, we show that the optimal scheme recommends the action $i$ maximizing $\spay{i} + \frac{n}{n-1}\rpay{i}$ (paying the optimal payments).\footnote{ Further recall that the sender/receiver payoffs for action $i$ are completely determined by action $i$'s type. So this can also be phrased as recommending a uniformly random action with type $k$, where $k$ maximizes $\xi_k + \frac{n}{n-1}\rho_k$ over all present types $k$.} When payments are non-negative, we prove that the optimal scheme is always either the optimal scheme without payments at all, or the optimal scheme with arbitrary payments.

\subsection{The general setting with payments}
\label{sec:payment-general}

Recall Equation~\ref{eq:Lagrange equation:main}, where we had the Langrangian function for an arbitrary polytope $\polytope$:
\begin{equation*}
\textstyle \lagrangeF_{\dualv} (\phi,P) = \sum_{\nstate\in\SSpace,i \in [n]}\dist\scheme{i} \left( \spay{i} + \rpay{i} \sum_{j \neq i} \dual{i}{j} - \sum_{j \neq i} \dual{i}{j} \rpay{j} \right) + \sum_{i\in[n]} P(i) \left( \sum_{j \neq i} \dual{i}{j} - 1 \right).
\end{equation*} 

%\[ \lagrangeF_{\lagrange}(\phi,P) = \sum_{\theta \in \Theta} \sum_{i \in [n]} \dist \scheme{i} \left( s(\theta,i) + \rpay{i} \sum_{j \neq i} \lagrange(i,j) - \sum_{j \neq i} \lagrange(i,j) \rpay{j} \right) + P(i) \left( \sum_{j \neq i} \lagrange(i,j) - 1 \right). \]

Also recall that for every choice $\lagrange$ of the Lagrange multipliers, $\max_{\phi,P}\lagrangeF_{\lagrange}(\phi,P)$ is an upper bound to the performance of the optimal persuasive scheme. Strong duality further implies that this bound is tight for some choice of the Lagrange multipliers.

Observe that if payments are allowed to be arbitrary, then $ \max_{\phi,P} \lagrangeF_{\lagrange}(\phi,P)$ is unbounded whenever the coefficient for $P(i)$ is non-zero for any $i$ (as we can simply set $P(i)$ to be $+\infty$ or $-\infty$. Therefore, we certainly have $\sum_{j \neq i} \lagrange(i,j) = 1$ in the optimal dual, for all actions $i$. This means that for each action $i$, the dual variables $\lagrange(i,.)$ form a distribution over actions other than $i$. The simplified Lagrangian becomes

\[ \lagrangeF_{\lagrange}(\phi,P) = \sum_{\theta \in \Theta} \sum_{i \in [n]} \dist \scheme{i} \left( \spay{i} + \rpay{i}  - \mathbb{E}_{j\sim\lagrange(i,.)} [ \rpay{j} ] \right). \]

For every choice of Lagrange multipliers $\lagrange(i,j)$, the scheme $\phi$ that maximizes $\lagrangeF_{\lagrange}(\phi,P)$ recommends, for every state of nature $\theta$, the action that maximizes $\spay{i} + \rpay{i}  - \mathbb{E}_{j\sim\lagrange(i,.)} [ \rpay{j} ]$.

\begin{observation}\label{obs:duality}
The optimal persuasive scheme $\phi$ recommends, at each state of nature $\theta$, the action that maximizes $s_\theta(i) + \rpay{i}  - \mathbb{E}_{j\sim\lagrange^*(i,.)} [ \rpay{j} ]$, where $\lagrange^*$ is the optimal Lagrange multiplier.
\end{observation}

Observation~\ref{obs:duality} provides a general framework to reason about optimal signaling schemes with arbitrary payments. We repeat now a connection to optimal auction design: In optimal auction design, there are some cases where the optimal dual is ``canonical,'' and doesn't depend on the input distribution (e.g. single-dimensional)~\citep{Myerson81}. In such settings, one can identify simple structure of the optimal mechanism. The case is similar in signaling: some cases admit a canonical optimal dual that doesn't depend on the input distribution. In these cases, we obtain simple characterizations of the optimal scheme. 
\subsection{Two actions, arbitrary payments}
\label{sec:binary general payments}
 In the $n=2$ case of only two actions, Observation~\ref{obs:duality} immediately allows us to derive the simple structure of the optimal signaling scheme.

\begin{proposition}\label{prop:general arbitrary binary actions}
When $n=2$, for every distribution $\dist$, the optimal persuasive scheme with possibly negative payments always recommends the action $i$ that maximizes $\spay{i} + 2\rpay{i}$ (and pays the optimal payments).
\end{proposition}

\begin{proof}
Observe that there are only two Lagrange multipliers, $\lagrange(0,1)$ and $\lagrange(1,0)$, and are both equal to $1$ in the optimal dual (by Observation~\ref{obs:duality}). Therefore, the Lagrangian can be further simplified:
\begin{talign*}
\lagrangeF_{\lagrange}(\phi,P) &= \sum_{\theta \in \Theta} \sum_{i \in [2]} \dist \scheme{i} \left( \spay{i} + \rpay{i} - \sum_{j \neq i} \rpay{j} \right) \\
&= \sum_{\theta \in \Theta} \sum_{i \in [2]} \dist \scheme{i} \left( \spay{i} + 2\rpay{i} - \sum_{j \in [2]} \rpay{j} \right).
\end{talign*}
Observe that the term $\sum_{j \in [2]} \rpay{j}$ does not depend on the action selected at all. So in order to maximize $\lagrangeF_{\lagrange}(\phi,P)$, the scheme must recommend the action maximizing $\spay{i} + 2\rpay{i}$ for every state of nature $\theta$. 
\end{proof}

\subsection{Symmetric actions}
Here, we draw conclusions for the symmetric setting with payments. Again getting initial traction from a canonical form for the optimal dual. Recall by Corollary~\ref{cor:symmetryLagrange}, the Lagrangian in this setting is given by:

\[ \textstyle \lagrangeF_{\lagrange}(\phi,P) = \sum_{\nstate\in\SSpace,i \in [n]}\dist\scheme{i} \left( \spay{i} + n \lagrange \rpay{i}  \right) + \sum_{i\in[n]} P(i) \left( (n-1) \lagrange - 1 \right) - \lagrange C. \]

\subsection*{Arbitrary payments}

When arbitrary payments are allowed, i.e. $\polytope = \mathbb{R}$, then the multiplier $(n-1)\lagrange - 1$ of the payment variable $P(i)$  must be equal to zero for all $i$. Otherwise the Lagrangian would be unbounded. This immediately implies that for the optimal dual, we have $\lagrange = \frac{1}{n-1}$, and the Lagrangian becomes

\[ \textstyle \lagrangeF_{\lagrange}(\phi,P) = \sum_{\nstate\in\SSpace,i \in [n]}\dist\scheme{i} \left( \spay{i} + \frac{n}{n-1} \rpay{i}  \right) - \frac{1}{n-1} C. \]

The proof of the following proposition then immediately follows.

\begin{proposition}
In the single sender, single receiver setting with symmetric actions and arbitrary payments, the optimal scheme recommends, on every state of nature $\nstate$, the action $i$ that maximizes $\spay{i} + \frac{n}{n-1} \rpay{i}$ (and pays the optimal payments). 
\end{proposition}

\subsection*{Non-negative payments: a dichotomy}

When payments are restricted to be non-negative, it is no longer the case that $\lagrange$ is pinned down completely (in particular, $\lagrange$ must certainly be $\leq \frac{1}{n-1}$, or else setting $P(i) = +\infty$ would result in an unbounded Lagrangian, but it is indeed possible to have $\lagrange < \frac{1}{n-1}$). Our next result shows that the optimal scheme in this scenario is essentially either the optimal scheme without any payments, or the optimal scheme for arbitrary payments (because the payments are already non-negative)

\begin{proposition}
In the single sender, single receiver setting with independent and identically distributed actions and non-negative payments, the optimal scheme is either (1) the optimal no-payment scheme, or (2) recommends the action $i$ that maximizes $\spay{i} + \frac{n}{n-1} \rpay{i}$ (and pays the optimal non-negative payments).
\end{proposition}

\begin{proof}
Let $\lagrange^* \in [0, \frac{1}{n-1}]$ be the $\lagrange$ guaranteed by Corollary~\ref{cor:symmetryLagrange}. Then there is an optimal scheme $(\phi^*, P^*)$ that maximizes $\lagrangeF_{\lagrange^*}(\phi, P)$ over all feasible $(\phi, P)$. 

If $\lagrange^* < \frac{1}{n-1}$, then $(n-1)\lagrange^* - 1$ is strictly negative, and hence the multiplier of each payment variable $P(i)$ is strictly negative. Therefore, every scheme that maximizes $\lagrangeF_{\lagrange^*}(\phi, P)$ must have $P(i) = 0$ for all actions $i$. Hence, the scheme $\phi$ is in fact feasible and persuasive for the no-payments case, and must be the optimal scheme without payments (as every scheme without payments is also feasible for non-negative payments, and $\phi$ is optimal among all schemes with non-negative payments). 

If $\lagrange^* = \frac{1}{n-1}$, then we immediately observe that this is exactly the same Lagrangian as for arbitrary payments, and therefore the second part of the proposition follows.
\end{proof}

\vspace{-5mm}
\section{Optimal Multi-agent Signaling with Externalities and Payments}
\label{sec:binary}
In this section, we study Bayesian persuasion for multiple receivers with binary actions and general externalities through the lens of duality. Our model is a natural extension of the model recently introduced in~\cite{AB-16}, and further developed in~\cite{DH-17}. We refer the reader to Section~\ref{sec:prelim} for notation and definitions.
Missing proofs can be found in Appendix~\ref{sec:missing-proofs}.
In Appendix~\ref{sec: computation plus externalities} we study optimal signaling with positive externalities: when no payments are allowed, and sender and receiver utility functions lie in some cone of set functions $\mathcal C$, we use duality to exhibit a polynomial time reduction from optimal signaling to the optimization problem for set functions in $\mathcal C$.

\subsection{Linear programming formulation}
\label{sec:binary-lp}
Similar to all other Bayesian persuasion problems, one can formulate finding the optimal signaling scheme with payments as a linear program, but this time with exponentially many variables (again, $\polytope$ is the set of feasible payments):
\begin{talign*}
\max~~~ &\sum_{\theta\in\SSpace}\sum_{S\subseteq [\agents]}\SDist\BScheme{S}\left(\Ffun{S}+\sum_{i\in S}\pay{i}{1}+\sum_{i\notin S}\pay{i}{0}\right)&\customlabel{eq:LP-binary-nopayments}{\textit{(LP-Binary)}}\\
&\sum_{\theta\in\SSpace}\SDist\sum_{S\ni i}\BScheme{S}\left(\util{i}{S}+\pay{i}{1}\right)\geq \sum_{\theta\in\SSpace}\SDist\sum_{S\ni i}\BScheme{S}\util{i}{S\setminus \{i\}}, & \forall i\in[\agents], \\
&\sum_{\theta\in\SSpace}\SDist\sum_{S\notni i}\BScheme{S}\left(\util{i}{S}+\pay{i}{0}\right)\geq \sum_{\theta\in\SSpace}\SDist\sum_{S\notni i}\BScheme{S}\util{i}{S\cup\{i\}}, & \forall i\in[\agents] \\
&\sum_{S\subseteq[\agents]}\BScheme{S}=1, ~~\forall \theta\in\SSpace~~~~~~~~~\BScheme{S}\geq 0, \forall S\subseteq [\agents], \theta\in\SSpace \\
&\{\pay{i}{a}\}_{i\in[\agents], a\in\{0,1\}}\in\polytope
\end{talign*}

The first two sets of constraints in this LP are essentially persuasion constraints, i.e. if a receiver is in the recommended set $S$ she is better off picking action $1$ and if a receiver is not in the recommended set $S$ she is better off picking action $0$. For notation brevity, we rewrite the persuasiveness constraints as:
\begin{align*}
\customlabel{eq:cons-1}{(*)}&~~~~~~~~~~\sum_{\theta\in\SSpace}\SDist\sum_{S\ni i}\BScheme{S}\left(\Gfun{i}{S}+\pay{i}{1}\right)\geq 0,&\forall i\in[\agents] \\
\customlabel{eq:cons-2}{(**)}&~~~~~~~~~~\sum_{\theta\in\SSpace}\SDist\sum_{S\notni i}\BScheme{S}\left(\Gfun{i}{S\cup\{i\}}-\pay{i}{0}\right)\leq 0,&\forall i\in[\agents]
\end{align*}

where $\Gfun{i}{S}\triangleq \util{i}{S}-\util{i}{S\setminus\{i\}}$. Note that $\Gfun{i}{S}=0$ for $i\notin S$.

\vspace{-3mm}
\subsection{Budget balanced payments under externalities }
As explored in Section~\ref{sec:payment}, adding monetary payments to the signaling problem, whether payments are positive or negative, is a natural idea to boost the performance of signaling schemes. It is expected that by adding arbitrary payments, one can increase the expected sender's utility, as we have seen in Section~\ref{sec:examples-contrast}; however, it is not clear what happens if we allow restricted payments, e.g. payments that are budget-balanced. A general and mathematically interesting setting to study this question is the general multi-agent signaling with externalities. We seek to find a simple characterization for the optimal scheme, and understand its properties. In the same setting, we also study arbitrary payments and show how one can modify the previous budget balanced scheme to find the optimal scheme with arbitrary payments. This model is again mathematically interesting, but perhaps not as natural as the budget-balanced.

Here is the surprising upshot of the story: with the help of Lagrangian duality we can design signaling schemes with payments that have a simple form (in contrast to the optimal scheme without payments), are budget-balanced in expectation (i.e. zero total payment in expectation),  and sender's expected payoff~\footnote{ When the scheme is budget-balanced, there is no difference between expected utility and expected payoff of the sender.} is no smaller than that of the optimal signaling without payments. Also, as we explained in Section~\ref{sec:examples-contrast}, somewhat surprisingly, it is possible that the sender's expected utility strictly increases by a budget-balanced scheme.
\begin{definition}
\label{def:budget-balance}
 A \emph{budget-balanced signaling scheme} for the multi-agent binary-actions with externalities setting is a pair of an \emph{allocation rule} $\{\BScheme{S}\}_{S\subseteq [\agents]}$ and a \emph{payment rule} $\{\pay{i}{a}\}_{i\in[\agents], a\in\{0,1\}}$ satisfying:
\begin{align*}
&\forall \theta\in \SSpace: \sum_{S\subseteq [\agents] }\BScheme{S}=1, ~~\forall S\subseteq[\agents],\theta\in\SSpace:\BScheme{S}\geq 0~~~~~&\textit{[feasibility]}\\
&\sum_{i\in [\agents]}\sum_{\theta\in\SSpace }\left(\sum_{S\ni i}\SDist\BScheme{S}\pay{i}{1}+\sum_{S\notni i}\SDist\BScheme{S}\pay{i}{0}\right)=0,~~~~~&\textit{[budget-balance]}
\end{align*}
where $\pay{i}{a}$ is the payment to receiver $i$ in state $\theta$, when the receiver takes action $a$.
\end{definition}

\begin{proposition}[Optimal budget-balanced signaling with externalities]\label{prop:optimal bb with externalities}
In multi-agent binary-action with externalities setting, there exists a persuasive signaling scheme such that:
\begin{enumerate}[leftmargin=*]
\item It is budget-balanced as in Definition~\ref{def:budget-balance} (but may charge negative payments),
\item It maximizes the ``total virtual payoff". That is, there exists a parameter $\gamma^*\geq 0$ such that for every state of nature $\theta$ the scheme recommends action $1$ to a subset of agents $S^*_\theta$ that maximizes $\Ffun{S}+\gamma^*\left(\sum_{i\in S}\Gfun{i}{S}-\sum_{i\notin S}\Gfun{i}{S\cup \{i\}}\right)$, where $\Gfun{i}{S}=\util{i}{S}-\util{i}{S\setminus \{i\}}$,
\item The expected total payoff of the sender is no smaller than the expected total payoff of the sender under any budget-balance persuasive scheme (with or without payments).
\end{enumerate}
\label{thm:budget-balance}
\end{proposition}

One can show a similar characterization for the optimal signaling scheme with arbitrary payments. In fact, because the payments are now unrestricted, the corresponding dual constraints are going to be equalities. This fact results in the following generalization of Proposition~\ref{prop:general arbitrary binary actions} for the special case of $N=1$.

\begin{proposition}[Optimal signaling with externalities and arbitrary payments]\label{prop:general-binary-external-arbitrary}
 In the multi-agent binary-action with externalities setting, the optimal signaling scheme with arbitrary payments maximizes the total payoff, i.e.  for every state $\theta$ the scheme recommends action $1$ to a subset of agents $S^*_\theta$ that maximizes
$\Ffun{S}+\left(\sum_{i\in S}\Gfun{i}{S}-\sum_{i\notin S}\Gfun{i}{S\cup \{i\}}\right)$, where $\Gfun{i}{S}=\util{i}{S}-\util{i}{S\setminus \{i\}}$.
\end{proposition}

\section{Conclusion}

We augment Bayesian persuasion by introducing payments, in single and multi-receiver settings. We obtain a number of results, all enabled via Lagrangian duality. For symmetric, single-receiver persuasion with no payments we show that a detail-free, ex-post Pareto optimal scheme is optimal. In the same setting, if arbitrary payments are allowed, the optimal scheme does not even depend on the prior: it always recommends the action that maximizes the sender utility plus $\frac{n}{n-1}$ times the receiver utility, where $n$ is the number of actions. When payments are restricted to be non-negative we prove a dichotomy: the optimal scheme is either the arbitrary-payment signaling scheme or the optimal no-payment scheme. When there are multiple receivers with binary actions and externalities, we prove that a simple scheme is optimal under a budget balanced constraint on the payments. Finally, in the same setting, when no payments are allowed, when the sender and receiver utility functions lie in some cone of set functions $\mathcal C$, we use duality to give a polynomial time reduction from optimal signaling to the optimization problem for set functions in $\mathcal C$.

Our work focuses on characterizing optimal schemes in the simplest cases beyond the tools of prior work. Notably, our tools gain the most traction where there is a ``canonical'' optimal dual (i.e. when the optimal Lagrangian multipliers are independent of the underlying distribution). Compare this to single-dimensional settings in auction design where the payment identity/monotonicity implies a canonical optimal dual as well~\cite{Myerson81}.

The clear direction for future work is to consider more general settings, possibly through the lens of ``simple versus optimal'' persuasion schemes (aligned with similar work in mechanism design). That is, rather than targeting optimal dual solutions (which are likely unwieldy too far beyond the canonical settings studied her), perhaps approximately optimal dual solutions will yield tractable insight for proving approximation results. Results of this form are limited without making use of duality, but do exist. For example,~\citet{DX-16} provides an \emph{even} simpler signaling scheme for a single receiver with a symmetric prior that guarantees a $(1-1/e)$-approximation. Another general direction for future work is to further explore the interplay between payments and persuasion.

% Bibliography
\bibliographystyle{plainnat}
\bibliography{refs}
% Appendix
\appendix
\section{Missing proofs.}
\label{sec:missing-proofs}
\begin{proof}[Proof of Proposition~\ref{prop:symmetry}]
Now, consider any optimal, persuasive signaling scheme $(\phi, p)$, and any optimal dual $\dualv$. For a permutation $\pi:[n]\rightarrow [n]$, let $\pi(\nstate)$ denote the state of nature $[\nstate_{\pi(1)},\ldots, \nstate_{\pi(n)}]$, $\pi(p)$ denote the prices $p(\pi^{-1}(1)),\ldots,p(\pi^{-1}(n))$. Let $\pi(\dualv)$ denote the dual variables with $\pi(\dualv)(i,j) = \dualv(\pi^{-1}(i),\pi^{-1}(j))$. Finally, let $\pi(\phi)$ denote the signaling scheme that on state of nature $\theta$, recommends action $i$ with probability $\phi_{\pi(\nstate)}(\pi^{-1}(i))$. 

Now we just want to confirm that $(\pi(\phi), \pi(P))$ and $\pi(\dualv)$ form another optimal primal/dual pair for any permutation $\pi$ (because $\dist$ is symmetric). Let's first compute the probability that $\pi(\phi)$ recommends action $\pi(i)$. By symmetry of $\dist$, the states $\nstate$ and $\pi^{-1}(\nstate)$ occur with the same probability. On state $\nstate$, $\phi$ recommends action $i$ with probability $\phi_{\nstate}(i)$. On state $\pi^{-1}(\nstate)$, $\pi(\phi)$ recommends action $\pi(i)$ with probability $\phi_{\pi(\pi^{-1}(\nstate))}(\pi^{-1}(\pi(i))) = \phi_{\nstate}(i)$ (this is just chasing through the definition of $\pi(\phi)$). Therefore, because $\nstate$ and $\pi^{-1}(\nstate)$ occur with the same probability, \emph{$\pi(\phi)$ recommends action $\pi(i)$ with the same probability $\phi$ recommends action $i$}. This immediately means that the total payment made in $\phi$ and $\pi(\phi)$ is identical because $(\phi,p)$ pays $p(i)$ whenever $\phi$ recommends action $i$, and $(\pi(\phi),\pi(p))$ pays $p(\pi^{-1}(\pi(i))) = p(i)$ whenever $\pi(\phi)$ recommends action $\pi(i)$, and these two probabilities are the same for all $i$.

Now we also want to claim that the above calculations show that the sender reward is identical under $\phi$ and $\pi(\phi)$ (assuming the receiver takes the recommended action). To see this, observe again that if we couple the events where the state of nature $\nstate$ is drawn for $\phi$, and state $\pi^{-1}(\nstate)$ is drawn for $\pi(\phi)$, then $\phi$ recommends action $i$ with type $\nstate_i$ with probability $\phi_{\nstate}(i)$, and $\pi(\phi)$ recommends action $\pi(i)$ with type $\nstate_{\pi^{-1}(\pi(i))} = \nstate_i$ with probability $\phi_{\nstate}(i)$. So in fact for every state of nature, the expected sender reward on state $\nstate$ under $\phi$ is the same as the expected sender reward on state $\pi^{-1}(\nstate)$ under $\pi(\phi)$. 

Next, we want to claim that the receiver's expected payoff for taking action $\pi(j)$ when $\pi(i)$ is recommended by $\pi(\phi)$ is exactly the same as their expected payoff for taking action $j$ when $i$ is recommended by $\phi$. We can write the expected payoff for taking action $j$ when $i$ is recommended (times the probability that $i$ is recommended by $\phi$) as:
$$\sum_{\nstate} \dist \phi_{\nstate}(i) \rpay{j} = \sum_{\nstate} \dist \pi(\phi)_{\pi^{-1}(\nstate)}(\pi(i)) r_{\pi(\nstate)}(\pi(j)).$$
The RHS now computes exactly the expected payoff for taking action $\pi(j)$ when $\pi(i)$ is recommended (times the probability that $\pi(i)$ is recommended by $\pi(\phi)$). So if $\phi$ is persuasive, then $\pi(\phi)$ is persuasive as well.

Finally, we want to claim that for the dual solution $\pi(\dualv)$, $(\pi(\phi),P)$ is optimal and satisfies complementary slackness. First, observe by the work above that the expected reward for the receiver for taking action $j$ when $i$ is recommended by $\phi$ is exactly the same as the expected reward for receiver for taking action $\pi(j)$ when $\pi(i)$ is recommended by $\pi(\phi)$. Combined with the fact that $\dualv(i,j) > 0 \Rightarrow$ the receiver is indifferent between following the recommendation and taking action $j$ when $\phi$ recommends $i$, and that $\pi(\dualv)(\pi(i),\pi(j)) = \dualv(\pi(\pi^{-1}(i)),\pi(\pi^{-1}(j)) )= \dualv(i,j)$, we immediately conclude that $\pi(\dualv)(\pi(i),\pi(j)) > 0 \Rightarrow$ the receiver is indifferent between following the recommendation and taking action $\pi(j)$ when $\pi(\phi)$ recommends $\pi(i)$. So complementary slackness is satisfied. Finally, observe that:
$$r_{\pi^{-1}(\theta)}^{\pi(\dualv)}(\pi(i)) = \sum_{j \neq i} \pi(\dualv)(\pi(i),\pi(j))\cdot (r_{\pi^{-1}(\theta)}(\pi(i)) - r_{\pi^{-1}(\theta)}(\pi(j)) )= \sum_{j \neq i} \dualv(i,j) \cdot (\rpay{i} - \rpay{j}) = \dualr{i}$$

So if on state $\nstate$, $i$ maximizes $\spay{i} + \dualr{i}$, then on state $\pi^{-1}(\nstate)$, $\pi(i)$ maximizes $\spay{\pi(i)} + r^{\pi(\dualv)}_{\pi^{-1}(\nstate)}(\pi(i))$, and $\pi(\phi)$ is indeed optimal for the Lagrangian problem induced by dual solution $\pi(\dualv)$. All together, this shows that $(\pi(\phi),\pi(P))$ is still persuasive and optimal, as is the dual $\pi(\dualv)$. We conclude by observing that the scheme that samples $\pi$ uniformly at random and then implements $(\pi(\phi),\pi(P))$ is therefore optimal and persuasive, as is the dual that averages $\pi(\dualv)$ over all $\pi$. It is easy to see that both primal and dual are symmetric as per the definitions.
\end{proof}

\begin{proof}[Proof of Proposition~\ref{prop:optimal bb with externalities}]
Similar to \ref{eq:LP-binary-nopayments}, we find the optimal scheme with payments through linear programming. 
We add the budget-balance constraint to the previous program. The optimal signaling scheme satisfying Definition~\ref{def:budget-balance} is the solution of this LP:
\begin{align*}
\max~~~ &\sum_{\theta\in\SSpace}\sum_{S\subseteq [\agents]}\SDist\BScheme{S}\Ffun{S} &\customlabel{eq:LP-binary-payments-1}{\textit{(LP-Budget-Balanced-1)}}\\
&\sum_{\theta\in\SSpace}\SDist\sum_{S\ni i}\BScheme{S}(\util{i}{S}+\pay{i}{1})\geq \sum_{\theta\in\SSpace}\SDist\sum_{S\ni i}\BScheme{S}\util{i}{S\setminus \{i\}}, & \forall i\in[\agents], \\
&\sum_{\theta\in\SSpace}\SDist\sum_{S\notni i}\BScheme{S}(\util{i}{S}+\pay{i}{0})\geq \sum_{\theta\in\SSpace}\SDist\sum_{S\notni i}\BScheme{S}\util{i}{S\cup\{i\}}, & \forall i\in[\agents] \\
&\sum_{i\in [\agents]}\sum_{\theta\in\SSpace }\left(\sum_{S\ni i}\SDist\BScheme{S}\pay{i}{1}+\sum_{S\notni i}\SDist\BScheme{S}\pay{i}{0}\right)=0&~\\
&\sum_{S\subseteq[\agents]}\BScheme{S}=1, ~~\forall \theta\in\SSpace~~~~~~~~~\BScheme{S}\geq 0, \forall S\subseteq [\agents], \theta\in\SSpace&
\end{align*}
For a given signaling scheme $\{\BScheme{S},\pay{i}{.}\}$, we introduce new variables $Q_i(1)$ and $Q_i(0)$ to be equal to the expected payment of agent $i$ for actions $1$ and $0$ respectively, where the expectation is taken over the randomness in the nature and the scheme, i.e.,
\begin{equation}
Q_i(1)\triangleq \sum_{\theta\in\SSpace }\sum_{S\ni i}\SDist\BScheme{S}\pay{i}{1}~~,~~Q_i(0)\triangleq \sum_{\theta\in\SSpace }\sum_{S\notni i}\SDist\BScheme{S}\pay{i}{0}
\end{equation}
We then simplify \ref{eq:LP-binary-payments-1} by rewriting it with variables $\{\BScheme{S},Q_i(.)\}$, i.e.,
\begin{align}
\max~~~ &\sum_{\theta\in\SSpace}\sum_{S\subseteq [\agents]}\SDist\BScheme{S}\Ffun{S}&\customlabel{eq:LP-binary-payments-2}{\textit{(LP-Budget-Balanced-2)}}\nonumber\\
&\sum_{\theta\in\SSpace}\SDist\sum_{S\ni i}\BScheme{S}\util{i}{S}+Q_i(1)\geq \sum_{\theta\in\SSpace}\SDist\sum_{S\ni i}\BScheme{S}\util{i}{S\setminus \{i\}},~~~~~~~~~\forall i\in[\agents]\label{eq:lp-binary-cons1}\\
&\sum_{\theta\in\SSpace}\SDist\sum_{S\notni i}\BScheme{S}\util{i}{S}+Q_i(0)\geq \sum_{\theta\in\SSpace}\SDist\sum_{S\notni i}\BScheme{S}\util{i}{S\cup\{i\}},~~~~~~~~~\forall i\in[\agents] \label{eq:lp-binary-cons2}\\
&\sum_{i\in [\agents]}\left(Q_i(1)+Q_i(0)\right)=0~\label{eq:lp-binary-cons3}\\
&\sum_{S\subseteq[\agents]}\BScheme{S}=1, ~~\forall \theta\in\SSpace~~~~~~~~~\BScheme{S}\geq 0, \forall S\subseteq [\agents], \theta\in\SSpace
\end{align}
To reveal the structure of the optimal scheme, we use the method of Lagrangian multipliers, \'a la Section~\ref{sec:nopayment}, and move the group of constraints \ref{eq:lp-binary-cons1},  \ref{eq:lp-binary-cons2} and \ref{eq:lp-binary-cons3} using dual variables $\{\alpha_i\}$,$\{\beta_i\}$ and $\gamma$ to the objective respectively. By rearranging the terms, this partial Lagrangian function $\mathcal{L}$ will be equal to 
\begin{align}
\mathcal{L}_{\alpha,\beta,\gamma}\left(\phi,Q(1),Q(0)\right)&=\ex{\theta}{\sum_{S}\BScheme{S}\Ffun{S}}\nonumber\\
&+\ex{\theta}{\sum_{i}\alpha_i\sum_{S\ni i}\BScheme{S}\Gfun{i}{S}}-\ex{\theta}{\sum_{i}\beta_i\sum_{S\notni i}\BScheme{S}\Gfun{i}{S\cup\{i\}}}\nonumber\\
&+\sum_{i}(\gamma-\alpha_i)Q_i(1)+\sum_{i}(\gamma-\beta_i)Q_i(0)\label{eq:lagrangian-binary}
\end{align}
Define $\mathcal{S}_\agents$ to be the simplex over all subsets of $[\agents]$. Strong duality implies that the optimal primal-dual solutions of \ref{eq:LP-binary-payments-2} are the solutions of the following min-max program:
\begin{multline*}
\max_{{Q}(1),{Q}(0)\in\mathbb{R}^n,\forall \theta:{\phi}_\theta\in\mathcal{S}_\agents}\left(\min_{{\alpha},{\beta}\in\mathbb{R}_+^n,\gamma\in\mathbb{R} }\mathcal{L}_{\alpha,\beta,\gamma}\left(\phi,Q(1),Q(0)\right)\right)\\
=\min_{{\alpha},{\beta}\in\mathbb{R}_+^n,\gamma\in\mathbb{R} }\left(\max_{{Q}(1),{Q}(0)\in\mathbb{R}^n,\forall \theta:{\phi}_\theta\in\mathcal{S}_\agents}\mathcal{L}_{\alpha,\beta,\gamma}\left(\phi,Q(1),Q(0)\right)\right)
\end{multline*}
By looking at the partial Lagrangian function $\mathcal{L}$ in \eqref{eq:lagrangian-binary}, $\forall i:\alpha^*_i=\beta^*_i=\gamma^*$, simply because otherwise one can make $Q_i(1)$ (or $Q_i(0)$) converging to either $+\infty$ or $-\infty$ to maximize the objective, and make the objective unbounded, contradicting the fact that the linear program is bounded. Therefore:

\begin{align}
\mathcal{L}_{\alpha^*,\beta^*,\gamma^*}&\left(\phi,Q(1),Q(0)\right)=\ex{\theta}{\sum_{S}\BScheme{S}\Ffun{S}} \\
&\qquad+\gamma^*\ex{\theta}{\sum_{i}\left(\sum_{S\ni i}\BScheme{S}\Gfun{i}{S}-\sum_{S\ni i}\BScheme{S}\Gfun{i}{S\cup\{i\}}\right)}\nonumber\\
&=\ex{\theta}{\sum_{S}\BScheme{S}\Ffun{S}}+\gamma^*\ex{\theta}{\sum_{S}\left(\sum_{i\in S}\BScheme{S}\Gfun{i}{S}-\sum_{i\notin S}\BScheme{S}\Gfun{i}{S\cup\{i\}}\right)}\nonumber\\
&=\ex{\theta}{\sum_{S}\BScheme{S}\left(\Ffun{S}+\gamma^*\left(\sum_{i\in S}\Gfun{i}{S}-\sum_{i\notin S}\Gfun{i}{S\cup\{i\}}\right) \right)}\label{eq:lagrangian-binary-after}
\end{align}
As the optimal signaling scheme ${\phi}^*$ should maximize \ref{eq:lagrangian-binary-after}, therefore for every state $\theta$ the scheme should recommend the set $S^*_\theta$ such that
\begin{equation*}
S_\theta^* = \underset{S}{\argmax}~\left(\Ffun{S}+\gamma^*\left(\sum_{i\in S}\Gfun{i}{S}-\sum_{i\notin S}\Gfun{i}{S\cup\{i\}}\right) \right)
\end{equation*}
To compute the payments, let $x^*_i=\pr{\theta}{i\in S^*_\theta}$. Now, to have a payment rule whose expectation is equal to $\mathbf{Q}(.)$, we can define the payments as
\begin{equation}
p^{*(i)}(1)=\frac{Q_i(1)}{x^*_i}\cdot\mathbb{I}\{i \in S_\theta^*\},~~p^{*(i)}(0)=\frac{Q_i(0)}{1-x^*_i}\cdot\mathbb{I}\{i \notin S_\theta^*\}
\end{equation}
Finally, it is clear that any signaling scheme without payments is a feasible solution for \ref{eq:LP-binary-payments-1}, and hence the sender's expected payoff in the above scheme (i.e. with allocation $\{S^*_\theta\}$ and payments ${p}^*(1),{p}^*(0)$)   is no smaller than the expected payoff of the sender in the optimal scheme without payments, which completes the proof.
\end{proof}

\newcommand{\conic}{\textsc{cone}}

\section{ Computing optimal signaling with positive externalities}\label{sec: computation plus externalities}
As the LP for binary signaling has exponentially many variables, one might wonder how hard it is to compute the optimal signaling scheme. In this section, we show a formal reduction from computing the optimal signaling scheme with positive externalities (and without payments\footnote{For the remaining of the section, we set the payments to be zero.}) to the optimization of a special class of set functions. Our reduction extends one direction of the reduction in~\cite{DH-17}, and uses techniques similar in spirit to~\cite{CDW-12,CDW-13,DW-15}.

The positive externality simply means that an agent switching from action $0$ to $1$ cannot harm any other agent. More formally, we have the following property.
 \begin{definition} 
 \label{def:positive}
 A profile of utility functions $\{\util{i}{.}\}$ has \emph{positive externalities} if and only if for every state of the nature $\nstate\in\SSpace$, and for every $S\subseteq[\agents]$, $i\in S$, and $j\neq i, j\in S$ we have: 
 \begin{align*}
 \Gfun{i}{S}\geq  \Gfun{i}{S\setminus \{j\}}
 \end{align*}
 where again, $\Gfun{i}{S}=\util{i}{S}-\util{i}{S\setminus \{i\} }$.
 \end{definition}

We quickly observe that under this property, it is in fact without loss to drop the constraints \ref{eq:cons-2} from the linear program \ref{eq:LP-binary-nopayments} in Section~\ref{sec:binary-lp}.

\begin{lemma}
\label{lem:drop-cons}
Suppose $\polytope=\vec{0}$. Under positive externalities,  there exists an optimal solution for the linear program~\ref{eq:LP-binary-nopayments} that also solves the same LP without constraints \ref{eq:cons-2}.
\end{lemma}
\begin{proof}
Consider any optimal solution to the LP without \ref{eq:cons-2} when $\polytope=\vec{0}$. If all constraints in \ref{eq:cons-2} happen to be satisfied anyway, then we're done. If not, we claim that we can satisfy \ref{eq:cons-2} without harming the quality of the solution or persuasiveness. Observe that if  \ref{eq:cons-2} is not satisfied by some optimal solution $\{\BScheme{S}\}$, then there exists $i,\theta,$ and $T$ such that $i\in T$, $\Gfun{i}{T}>0$ and $\BScheme{T\setminus\{i\}}>0$. Consider modifying $\{\BScheme{S}\}$ by moving all mass from $\BScheme{T\setminus\{i\}}$ to $\BScheme{T}$ instead. Observe that:
\begin{itemize}
\item This only makes the sender weakly happier, as the sender's payoff function is monotone. 
\item This only makes receiver $i$ strictly happier, as we had $\Gfun{i}{T} > 0$. 
\item If $j\neq i, j\in T$ was recommended action $1$, then this recommendation still remains persuasive, due to the marginal cross-monotonicity property.
\item if $j\neq i, j\in T$ was recommended action $0$ and this action is not persuasive anymore, modify $\{\BScheme{S}\}$ by recommending action $1$ to her instead.
\end{itemize}
Now, repeat the above process until there is no constraint in \ref{eq:cons-2} violated. This process terminates in finite time, because once a person is recommended action $1$ this recommendation remains persuasive until the end, and at each iteration we either terminate or make progress by recommending action $1$ to least one more person. Furthermore, this process only makes the sender's expected payoff higher, completing the proof.
\end{proof}

Applying Lemma~\ref{lem:drop-cons}, the final primal-dual LP (without payments) can be simplified as:
\begin{align*}
\max &\sum_{\theta\in\SSpace}\sum_{S\subseteq [\agents]}\SDist\BScheme{S}\Ffun{S}
&\min &\sum_{\theta\in\SSpace}y_\theta \\
&\sum_{\theta\in\SSpace}\SDist\sum_{S\ni i}\BScheme{S}\Gfun{i}{S}\geq 0,~~~\forall i\in[\agents]
&~~&y_\theta-\sum_{i\in S}\alpha_i\Gfun{i}{S}\SDist\geq \Ffun{S}\SDist,&\forall S\subseteq [\agents], \theta\in\SSpace  \\
&\sum_{S\subseteq[\agents]}\BScheme{S}=1, ~~~~~~~~~~~~~~~~~~~~~\forall \theta\in\SSpace
&~~&\alpha_i\geq 0,&\forall i\in[\agents]  \\
&\BScheme{S}\geq 0, ~~~~~~~~~~~~~~~~~\forall S\subseteq [\agents], \theta\in\SSpace
&~~&&
\end{align*}

We also need to define \emph{convex cones} over the space of set functions before describing our result.
\begin{definition}
Any subset $\mathcal{C}$ of $\mathbb{R}^\Omega$, where $\Omega\triangleq 2^{[N]}$, is a \emph{cone} if and only if for each set function $h\in\mathcal{C}$ and scalar $\alpha\in\mathbb{R}_{+}$, $\alpha\cdot h\in\mathcal{C}$. Moreover, given set functions $f_1(.),\ldots,f_m(.)$, the \emph{conic hull} of these functions is defined as $\conic(f_1,\ldots,f_m)\triangleq \{f:2^{[\agents]}\rightarrow \mathbb{R}: f=\sum_{i=1}^{m}w_i\cdot f_i, w_i\geq 0\}$, which indeed is a cone in $\mathbb{R}^\Omega$.
\end{definition}
%\begin{definition}
%Given set functions $\mathcal{T}=(f,g_1,\ldots,g_\agents)$, the \emph{viable signaling instances of $\mathcal{T}$}, denoted by $\mathcal{I}(\mathcal{T})$, are instances of binary signaling with externalities such that  $f^\theta=f$ for every state $\theta\in \SSpace$, and $\exists~a^{\theta}_{i},b_i^{\theta}\in \reals$  such that $g_i^{\theta}=a^{\theta}_{i}+b_i^{\theta}g_i$ for every state $\theta\in \SSpace$.
%\end{definition}

Given these definitions, we formally prove the following proposition.

\begin{proposition} [BP with externalities~$\Rightarrow$~Set optimization  ]
\label{thm:opt-to-sig}
Let $\mathcal{C}\subseteq {\mathbb{R}^\Omega}$ be a cone of set functions, and suppose there is polynomial time algorithm that returns $\argmax_{S\subseteq [N]}h(S)$ for every set function $h\in \mathcal{C}$. Then, there exists a polynomial time algorithm that computes the optimal signaling scheme for every instance of Bayesian persuasion with externalities satisfying the following:
\begin{enumerate}
\item For every state of the nature $\theta\in \SSpace$, $f^{\theta}\in\mathcal{C},~g_i^{\theta}\in\mathcal{C}$ for all $i\in[N]$, and
\item For every state of the nature $\theta\in \SSpace$, $f^\theta$ is monotone non-decreasing and agents have positive externalities as in Definition~\ref{def:positive}.
\end{enumerate}
\end{proposition}

%\begin{proposition}[BP with externalities~$\Rightarrow$~Set optimization  ]
%Let  $\Ffun{S}$ be the payoff function of the sender and $\{\Gfun{i}{\cdot}\}$ be the marginal payoff functions of the receivers, and assume for every state of the nature $\theta$,  $\Ffun{S}$ is monotone non-decreasing and receivers have positive externalities. Given access to a polynomial time algorithm for optimizing any set function in $\conic (f^\theta,g_1^\theta,\ldots,g_{\agents}^\theta)$ for any state of , there is a polynomial time algorithm for computing the optimal signaling scheme.
%\end{proposition} 
%\begin{theorem} [Set optimization $\leq_p$ Optimal signaling ] 
%\label{thm:sig-to-opt}
%Given a monotone set function $f$ and set functions $g_1,\ldots,g_{\agents}$, there is a polynomial time algorithm for maximizing any set function in $\conic (f, g_1,\ldots,g_{\agents})$ if there is a polynomial time algorithm for computing the optimal signaling scheme of any instance in $\mathcal{I}(f, g_1,\ldots,g_{\agents})$.
%\end{theorem} 

\begin{proof}
Consider the primal linear program of optimal signaling. By using linear programming duality, solving the primal LP is reduced to solving its dual. Moreover, because of the equivalence between optimization and separation, solving the dual program is reduced to finding a separation oracle for the following set of dual constraints:
$$
y_\theta-\sum_{i\in S}\alpha_i\Gfun{i}{S}\SDist\geq \Ffun{S}\SDist,~~~\forall S\subseteq [\agents], \theta\in\SSpace
$$
Clearly $\Gfun{i}{S}=\util{i}{S}-\util{i}{S\setminus\{i\}}=0$ for $i\notin S$, and hence $\sum_{i\in S}\alpha_i\Gfun{i}{S}=\sum_{i\in [\agents]}\alpha_i\Gfun{i}{S}$. By dividing both sides of the constraint by $\SDist$, the separation problem for every $\theta\in\SSpace$ is equivalent to finding $S^*=\argmax_{S\subseteq[\agents]} \sum_{i\in [\agents]}\alpha_i\Gfun{i}{S}+  \Ffun{S}$ for given $\alpha_i$'s. Also, $\alpha_i\geq 0$, and hence $\sum_{i\in [\agents]}\alpha_i\Gfun{i}{S}+  \Ffun{S}\in \conic({f^\theta,g_1^{\theta},\ldots,g_{\agents}^\theta}) \subseteq \mathcal{C}$. So, given access to the algorithm that solves the set function optimization over $\mathcal{C}$ in polynomial time, we can find $S^*$ in polynomial time. Given this separation oracle, we can now solve the persuasion problem in polynomial time.
\end{proof}

\section{Reduced forms in i.i.d. signaling}
\label{appendix:reduced form}
By making use of ideas related to reduced forms in Bayesian auctions and Border's theorem~\cite{B-91,M-84,MR-84}, we can simplify the LP for the independent setting. These ideas have already appeared in~\cite{DX-16}, but we reproduce them here for completeness. Let $B^\nstate \in \{ 0,1 \}^{n \times m}$ be an $n$ by $m$ matrix such that $B_{ij}^\nstate = 1$ iff $\nstate_i = j$ (action $i$ has type $j$). Define $\signi{i} = \sum_{\nstate \in \SSpace}\dist \scheme{i} B^{\theta}$, that is $\signijk{i}{j}{k}$ is the joint probability that action $j$ has type $k$ and the scheme outputs action $i$.
$\signv = \left( \signi{1}, \dots, \signi{n} \right) \in \mathbb{R}^{n \times m \times n}$ is the \emph{signature} or \emph{reduced form} of $\schemev$.
$\signv$ is realizable if there exists a signaling scheme $\schemev$ with $\signv$ as its signature. We write $\spolytope$ for the polytope of realizable signatures (see \cite{DX-16} for more details and properties of these reduced forms). The program for this special case, using the same change of variables as in the general case, is:
\begin{align*}
\max~~~&\sum_{i \in [n]} \sum_{k \in [m]} \signijk{i}{i}{k}\sipay{i}{k} -\sum_{i\in [n]}P(i)&\customlabel{eq:Lp-independent-payments}{\textit{(LP-Independent-with-Payments)}}\\
&P(i)+ \sum_{k \in [m]}  \signijk{i}{i}{k}\ripay{i}{k}  \geq \sum_{k \in [m]}  \signijk{i}{j}{k}\ripay{i}{k}   , & \text{ for } i,j \neq i \in [n] \\
&\signv \in \spolytope, ~~~\text{and}~~~P \in\polytope &
\end{align*}
While these notations are helpful to identify the computational hardness of signaling problems, as in \cite{DH-17,DX-16}, we never use them in our treatment, and we only mention them here for completeness.

\end{document}